\documentclass[12pt,draftcls,onecolumn]{IEEEtran}
\usepackage[update,prepend]{epstopdf}
\usepackage{amsmath,rotating}
\usepackage{hyperref}
\usepackage{amsfonts}
\usepackage{amssymb}
\usepackage{times,subfigure}
\usepackage{comment}
\usepackage{epstopdf}
\usepackage{algorithmic}
\usepackage{algorithm}
\usepackage{graphicx}
\usepackage[T1]{fontenc}
\usepackage[scanall]{psfrag}
\newtheorem{Theorem}{Theorem}
\newtheorem{Corollary}{Corollary}
\newtheorem{Lemma}{Lemma}

\newtheorem{Remark}{Remark}

\def\sqw{\hfill\hbox{\lower.1ex\hbox{$\sqcup$}
    \kern-1.02em\lower.1ex\hbox{$\sqcap$}}\ }

\newcommand{\qed}{\hfill \mbox{\raggedright \rule{.07in}{.1in}}}

\newenvironment{proof}{\vspace{1ex}\noindent{\bf Proof}\hspace{0.5em}}{\hfill\qed\vspace{1ex}}

\newcommand{\yv}{\mathbf{y}}
\newcommand{\av}{\mathbf{a}}
\newcommand{\bv}{\mathbf{b}}

\newcommand{\xv}{\mathbf{x}}

\newcommand{\vv}{\mathbf{v}}

\newcommand{\AM}{\mathbf{A}}
\newcommand{\BM}{\mathbf{B}}

\newcommand{\PM}{\mathbf{P}}

\newcommand{\IM}{\mathbf{I}}
\newcommand{\MM}{\mathbf{M}}
\newcommand{\NM}{\mathbf{N}}

\newcommand{\R}{\mathbb{R}}

\title{On the exponential convergence of\\ the Kaczmarz algorithm}

\author{
\IEEEauthorblockN{
Liang Dai and Thomas B. Sch\"{o}n\\
\IEEEauthorblockA{Department of Information Technology, Uppsala University,\\ 751 05 Uppsala, Sweden. E-mail: \{liang.dai, thomas.schon\}@it.uu.se}
\thanks{This work was supported by the project \emph{Probabilistic modelling of dynamical systems} (Contract number: 621-2013-5524) funded by the Swedish Research Council.}}
}

\begin{document}

\maketitle
\begin{abstract}
    The Kaczmarz algorithm (KA) is a popular method for solving a system of linear equations. In
    this note we derive a new exponential convergence result for the KA. The key allowing us to
    establish the new result is to rewrite the KA in such a way that its solution path can be
    interpreted as the output from a particular dynamical system. The asymptotic stability results
    of the corresponding dynamical system can then be leveraged to prove exponential convergence of
    the KA. The new bound is also compared to existing bounds.
\end{abstract}


\IEEEpeerreviewmaketitle

\begin{section}{Problem Statement}
In this note, we discuss the exponential convergence property of the Kaczmarz algorithm (KA) \cite{c4}. Since its introduction, the KA has been applied in many different fields and many new developments are reported\cite{c8}-\cite{c19}. The KA is used to find the solution to the following system of
{\em consistent} linear equations
\begin{align}
\AM\xv = \bv,
\end{align}
where $\xv \in \mathbb{R}^{n}$ denotes the unknown vector, $\AM \in \mathbb{R}^{m\times n}, m\ge n$, $\text{rank}(\AM) = n$ and $\bv \in \mathbb{R}^{m}$. Define the hyperplane $H_i$ as
\begin{align*}
H_i = \{\xv| \av_i^{T}\xv = b_i\},
\end{align*}
where the $i$-th row of $\AM$ is denoted by $\av_i^{T}$ and the $i$-th element of $\bv$ is denoted by $b_i$. Geometrically, the KA finds the solution by projecting (or approximately projecting) onto the hyperplanes cyclically from an initial approximation $\xv_0$, which reads as
\begin{align}
\xv_{k+1} = \xv_k + \lambda \frac{b_{i(k)}-\av_{i(k)}^{T}\xv_k}{\|\av_{i(k)}\|_2^2}\av_{i(k)},
\label{eq.rec}
\end{align}
where $ i(k) = \text{mod}(k,m) +1.$
 In the update equation~\eqref{eq.rec}, $\lambda$ is the relaxation parameter, which satisfies $0< \lambda < 2$. We use the Matlab convention $\text{mod}(\cdot,\cdot)$ to denote the {\em modulus after division} operation and $\|\cdot\|_2$ to denote the spectral norm of a matrix.

 It is well-known that the KA is sometimes rather slow to converge. This is especially true when
 several consecutive row vectors of the matrix $\AM$ are in some sense "close" to each other. In
 order to overcome this drawback, the Randomized Karczmarz Algorithm (RKA) algorithm was introduced
 in~\cite{c10} for $\lambda =1$. The key of the RKA is that, instead of performing the hyperplane projections
 cyclically in a deterministic order, the projections are performed in a random order. More
 specifically, at time $k$, select a hyperplane $H_p$ to project with probability
 $\frac{\|\av_p\|_2^2}{\|\AM\|_{F}^2}$, for $p = 1,\cdots,m$. Note that $\|\cdot\|_{F}$ is used to denote
 the Frobenius norm of a matrix. Intuitively speaking, the involved randomization is performing a
 kind of "preconditioning" to the original matrix equations~\cite{c12}, resulting in a faster
 exponential convergence rate, as established in~\cite{c10}.

 The specific and predefined ordering of the projections in the KA makes it challenging to obtain a
 tight bound of the convergence rate of the method. In~\cite{c22}, the authors build up the
 convergence rate of the KA by exploiting the Meany inequality~\cite{c23}, which works for the case
 $\lambda=1$ in~\eqref{eq.rec}. ~\cite{c24,c26} also established convergence rates for
 the KA, for $\lambda \in (0,2)$. In Section~\ref{sec:DiscNumIll}, we will compare
 these results in more detail.

In this note, we present a different way to characterize the convergence property of the KA described in
\eqref{eq.rec}. The key underlying our approach is that we
interpret the solution path of the KA as the output of a particular dynamical system. By studying
the stability property of this related dynamical system, we then obtain new exponential convergence
results for the KA. Related to this, it is interesting to note that the so-called Integral Quadratic
Constraints (IQCs) has recently been used in studying the convergence rate of
first-order algorithms applied to solve general convex optimization problems~\cite{c25}.

The note will be organized as follows. In the subsequent section we make use of the sub-sequence
$\{\xv_{jm}-\xv\}_{j=0}^{\infty}$ to enable the derivation of the new exponential convergence
result. In Section~\ref{sec:DiscNumIll} we discuss its connections and differences to existing
results. Conclusions and ideas for future work are provided in Section~\ref{sec:Conc}.
\end{section}

\begin{section}{The new convergence result}
First, let us introduce the matrix $\BM\in \R^{m\times n}$, for which the $i$-th row $\bv_i^{T}$ is defined as $\bv_i \triangleq \frac{\av_i}{\|\av_i\|_2}, i = 1,\cdots,m$. Furthermore, let $\PM_i \triangleq \bv_i\bv_i^{T}$ for $i = 1,2,\cdots,m$ and let $\theta_{k} \triangleq \xv_k - \xv$ for $k\ge0$. Using this new notation allows us to rewrite~\eqref{eq.rec} according to
\begin{align}
\theta_{k+1} = (\IM - \lambda \PM_{i(k)})\theta_k,
\label{eq.dy}
\end{align}
which can be interpreted as a discrete time-varying linear dynamical system. Hence, this relation inspires us to study the KA by employing the techniques for analyzing the stability properties of time-varying linear systems, see e.g.~\cite{c20,c21}.

In what follows we will focus on analyzing the convergence rate of the sub-sequence $\{\|\theta_{jm}\|^2\}_{j=0}^{\infty}$. Given the fact that $ i(k) = \text{mod}(k,m) +1$, we have
\begin{align*}
\theta_{(j+1)m} = \left( \prod_{i=1}^{m}{(\IM - \lambda \PM_i)}\right) \theta_{jm}\triangleq \MM_m \theta_{jm}.
\end{align*}

The following theorem provides an upper bound on the spectral norm of $\MM_m$.
\begin{Theorem}
Let $\rho \triangleq \|\MM_m\|_2$ and $0< \lambda \le 2$, then it holds that
\begin{align}
\rho^2 \le \rho_1 \triangleq  1-\frac{\lambda(2-\lambda)}{(2+\lambda^2m^2)\|\BM^{\dagger}\|_2^2},
\label{eq.bd}
\end{align}
where $\BM^{\dagger}$ denotes the pseudo-inverse of the matrix $\BM$.
\end{Theorem}

\begin{proof}
Let $\vv_0 \in \R^{n}$ be a vector satisfying $\MM_m\vv_0 = \rho\vv_0$, $\|\vv_0\|_2 = 1$ and let $\vv_{i}= (\IM - \lambda\PM_{i})\vv_{i-1}$ for $i = 1,\cdots,m$. It follows that $\vv_m = \MM_m \vv_0$ and $\|\vv_m\|^2 = \rho^2$.

Notice that $\PM_i^2 = \PM_i$, so we have
\begin{align*}
(\IM-\lambda \PM_i)^2 &= \IM - (2\lambda-\lambda^2)\PM_i,
\end{align*}
for $i = 1,\cdots,m$.
Hence it holds that
\begin{align*}
\|\vv_i\|^2 &= \vv_{i-1}^{T}(\IM-\lambda \PM_i)^2\vv_{i-1} \\
& = \vv_{i-1}^{T}(\IM-\lambda(2-\lambda) \PM_i)\vv_{i-1} \\
& = \|\vv_{i-1}\|^2 - \lambda(2-\lambda)\|\PM_i\vv_{i-1}\|^2,
\end{align*}
which in turn implies that
\begin{align}
\lambda(2-\lambda)\sum_{i=1}^{m}\|\PM_i\vv_{i-1}\|^2 = \|\vv_0\|^2 - \|\vv_m\|^2 = 1- \rho^2.
\label{eq.bd1}
\end{align}
Also, for any $i \in \{1,\cdots,m\}$, we have that
\begin{align*}
&\|\vv_{i}-\vv_{0}\| \\
&= \left\| \sum_{k=1}^{i}(\vv_{k}-\vv_{k-1})\right\| = \lambda \left\|\sum_{k=1}^{i}\PM_k\vv_{k-1}\right\| \\
& \le \lambda \sum_{k=1}^{i}\left\|\PM_k \vv_{k-1}\right\| \le \lambda \sqrt{i} \sqrt{\sum_{k=1}^{i}\left\|\PM_k \vv_{k-1}\right\|^2}\\
& \le  \sqrt{\lambda i} \sqrt{\lambda\sum_{k=1}^{m}\left\|\PM_k \vv_{k-1}\right\|^2}
\end{align*}
Together with~\eqref{eq.bd1}, we get
\begin{align}
\left\|\vv_{i}-\vv_{0}\right\|^2 \le \frac{\lambda i}{2-\lambda}(1-\rho^2).
\label{eq.bd2}
\end{align}
Meanwhile, we have that
\begin{align}
&\lambda\vv_{0}^{T}\BM^{T}\BM\vv_{0} \nonumber\\
&= \lambda \sum_{k=1}^{m}\vv_{0}^{T}\PM_k\vv_{0} = \lambda \sum_{k=1}^{m}\|\PM_k\vv_{0}\|^2 \nonumber\\
& =\lambda \sum_{k=1}^{m}\|\PM_k[\vv_{k-1} + (\vv_{0}-\vv_{k-1})]\|^2 \nonumber\\
&\le 2\lambda \sum_{k=1}^{m}\|\PM_k\vv_{k-1}\|^2 + 2\lambda \sum_{k=1}^{m}\|\PM_k(\vv_{k-1}- \vv_{0})\|^2 \nonumber\\
& \le  2\lambda \sum_{k=1}^{m}\|\PM_k\vv_{k-1}\|^2 + 2\lambda \sum_{k=1}^{m}\|\vv_{k-1}- \vv_{0}\|^2.
\label{eq.loose}
\end{align}
Together with (\ref{eq.bd1}) and (\ref{eq.bd2}), we have that
\begin{align*}
\lambda\vv_{0}^{T}\BM^{T}\BM\vv_{0} \le \frac{2(1-\rho^2)}{2-\lambda} + 2\lambda \sum_{k=1}^{m}\frac{\lambda (k-1)}{2-\lambda} (1-\rho^2)
\end{align*}
or equivalently
\begin{align}
\lambda\vv_{0}^{T}\BM^{T}\BM\vv_{0}  \le \frac{1-\rho^2}{2-\lambda}\left(2+\lambda^2m(m-1)\right),
\end{align}
hence it follows that
\begin{align*}
\rho^2 \le 1-\frac{\lambda(2-\lambda)\vv_{0}^{T}\BM^{T}\BM\vv_{0} }{2+\lambda^2m(m-1)}.
\end{align*}
Since $\vv_{0}^{T}\BM^{T}\BM\vv_{0} \ge\frac{1}{\|\BM^{\dagger}\|_2^2} $, we conclude that
\begin{align}
  \label{eq:rho2}
  \rho^2 \le  1-\frac{\lambda(2-\lambda)}{\left(2+\lambda^2m(m-1)\right)\|\BM^{\dagger}\|_2^2}.
\end{align}

Finally, notice that $m(m-1) \le m^2$ holds for any natural number $m$, which concludes the proof.
\end{proof}

%

\vspace{+1mm}
\begin{Remark}
  Notice that in the proof of Theorem~1, \eqref{eq.loose} is the main approximation step , and a better approximation here will lead to an improvement of the bound.
\end{Remark}

\vspace{+3mm}
The following corollary characterizes the convergence of the KA under $\lambda=1$, which will be
used in the subsequent section to enable comparison to the results given in \cite{c22,c23}. We omit the proof since
it is a direct implication of Theorem 1.

\vspace{+1mm}
\begin{Corollary}
For the KA with $\lambda=1$ in~\eqref{eq.rec}, if $m\ge n\ge2$, we have that
\begin{align}
\rho^2 \le 1-\frac{1}{2m^2\|\BM^{\dagger}\|_2^2}.
\label{eq.comp1}
\end{align}
\end{Corollary}

Next, we will derive an improvement over the bound (\ref{eq.bd}), enabled by partitioning the matrix $\AM$ into non-overlapping sub-matrices. Let $q = \lceil\frac{m}{n}\rceil + 1$, where $\lceil x \rceil$ denotes the smallest number which is greater or equal to $x$. Define the following sets as $T_i = \{ (i-1)n+1, \cdots, i n\}$, for $i = 1, \cdots, q-1$ and $T_q = \{ (q-1)n+1, \cdots, m\}$.  Further, for $i = 1,\cdots, q$,  define $\BM_i$ as the sub-matrix of $\BM$ with the rows indexed by the set $T_i$, and $\NM_i = \prod_{j\in T_i}(\IM - \lambda \PM_j)$.

\begin{Corollary}
Based on the previous definitions, and further assume that all the sub-matrices $\BM_i$ for for $i = 1,\cdots, q$ are of rank $n$, then we have that
\begin{align}
\rho^2 \le \rho_2 \triangleq \prod_{i=1}^{q}\left(1 - \frac{\lambda(2-\lambda)}{\left(2+\lambda^2n(n-1)\right)\|\BM_i^{\dagger}\|_2^2}\right)
\label{eq.prodbd}
\end{align}
\end{Corollary}
\begin{proof}
Notice that since
\begin{align*}
\MM_m = \NM_q\NM_{q-1}\cdots\NM_2\NM_1,
\end{align*}
we have that
\begin{align}
\rho^2 = \| \MM_m\|_2^2 \le \prod_{i=1}^{q} \| \NM_i\|_2^2.
\label{eq.prod}
\end{align}

For each $\NM_i$, the spectral norm can be bounded analogously to what was done in Theorem~1, resulting in
\begin{align*}
\| \NM_i\|_2^2 \le 1 - \frac{\lambda(2-\lambda)}{\left(2+\lambda^2n(n-1)\right)\|\BM_i^{\dagger}\|_2^2}.
\end{align*}

Finally, inserting this inequality into (\ref{eq.prod}) concludes the proof.
\end{proof}

\end{section}

\begin{section}{Discussion and numerical illustration}
\label{sec:DiscNumIll}
In Section~\ref{sec:Disc:Meany} and~\ref{sec:rka} we compare our new bound with the bounds provided by the Meany
inequality \cite{c22,c23} and the RKA, respectively. In Section~\ref{sec:rka} we also provide a
numerical illustration. Section~\ref{sec:num} is devoted to a comparison with the bound provided in~\cite{c24}, and finally Section~\ref{sec:nbd} compares with the result given by~\cite{c26}.

\subsection{Comparison with the bound given by Meany inequality}
\label{sec:Disc:Meany}
In the following, we assume that $m=n$ and $\lambda =1$. Denote the singular values of $\BM$ as $\sigma_1 \ge \sigma_2 \cdots \ge \sigma_n$, then the bound in \cite{c22,c23} given by the Meany inequality can be written as $\rho^2 \le 1- \prod_{i=1}^{n}\sigma_i^2 $, and the bound given in (\ref{eq.comp1}) can be written as $\rho^2 \le 1- \frac{\sigma_n^2}{2n^2}$. This implies that when
\begin{align}
\frac{\sigma_n^2}{2n^2} \ge \prod_{i=1}^{n}\sigma_i^2 \,\,\, \text{i.e. }\,\,\, \prod_{i=1}^{n-1}\sigma_i^2 \le \frac{1}{2n^2},
\label{condd}
\end{align}
holds, the bound in (\ref{eq.comp1}) is tighter. In the following lemma, we derive a sufficient condition, under which the inequality $\eqref{condd}$ holds.

\begin{Lemma}
\label{lem:Meany}
If $\sigma_{n-1}^2 \le \frac{(n-2)^{n-2}}{2n^n}$ holds, the inequality in~\eqref{condd} is satisfied.
\end{Lemma}
\begin{proof}
Notice that
\begin{align}
\prod_{i=1}^{n-1}\sigma_i^2 = \left(\prod_{i=1}^{n-2}\sigma_i^2\right) \sigma_{n-1}^2 & \le \left(\frac{\sum_{i=1}^{n-2}\sigma_i^2}{n-2}\right)^{n-2}\sigma_{n-1}^2 \nonumber \\
&  \le \left(\frac{n}{n-2}\right)^{n-2}\sigma_{n-1}^2. \label{ine} 
\end{align}
The inequality~\eqref{ine} holds since
\begin{align*}
\sum_{i=1}^{n-2}\sigma_i^2 \le \sum_{i=1}^{n}\sigma_i^2 = \|\BM\|_{F}^2 = n.
\end{align*}
Hence, if
\begin{align*}
\left(\frac{n}{n-2}\right)^{n-2}\sigma_{n-1}^2 \le \frac{1}{2n^2}
\end{align*}
holds, or equivalently if
\begin{align*}
\sigma_{n-1}^2  \le \frac{(n-2)^{n-2}}{2n^n}
\end{align*}
holds, then  $\eqref{condd}$ holds, which concludes the proof.
\end{proof}
\begin{Remark}
  Notice that the right hand side of the inequality in Lemma~\ref{lem:Meany} is in the order of
  $\frac{1}{n^2}$ for large $n$. Another difference is that the bound provided by Theorem 1 depends explicitly on
  the size of matrix, while the bound provided by the Meany inequality does not.
\end{Remark}

\subsection{Comparison with the bound given by the RKA}
\label{sec:rka}
Let us now compare our new results to the results available for the RKA. Note that in this case, we set $\lambda =1$. If
$\{\theta_{jm}\}_{j=0}^{\infty}$ denotes the sequence generated by the RKA, then it holds
that~\cite{c10}
 \begin{align}
\mathbb{E}\|\theta_{jm}\|^2 \le \left( 1-\frac{1}{\|\AM\|_{F}^2\|\AM^{\dagger}\|_2^2} \right)^{jm}\|\theta_{0}\|^2,
\label{eq.comp2}
\end{align}
for $j\ge 1$, where $\mathbb{E}$ denotes the expectation operator with respect to the random operations up to index $jm$.

To compare (\ref{eq.comp2}) and (\ref{eq.comp1}), we make the assumption that $\AM$ is a matrix with each row normalized, i.e. $\AM = \BM$, for simplicity. It follows that $\|\BM\|_F^2=m\le m^2$, and
\begin{align}
1-\frac{1}{2m^2\|\BM^{\dagger}\|_2^2} \ge 1-\frac{1}{\|\BM\|_F^2\|\BM^{\dagger}\|_2^2}.
\label{eq.ineq1}
\end{align}
Furthermore, since $\|\BM\|_F^2\|\BM^{\dagger}\|_2^2 \ge 1$, we have that
\begin{align}
1-\frac{1}{\|\BM\|_F^2\|\BM^{\dagger}\|_2^2}\ge \left(1-\frac{1}{\|\BM\|_F^2\|\BM^{\dagger}\|_2^2}\right)^{m},
\label{eq.ineq2}
\end{align}
and combining (\ref{eq.ineq1}) and (\ref{eq.ineq2}), results in
\begin{align*}
1-\frac{1}{2m^2\|\BM^{\dagger}\|_2^2} \ge \left(1-\frac{1}{\|\BM\|_F^2\|\BM^{\dagger}\|_2^2}\right)^{m}.
\end{align*}
The above inequality implies that the bound given by \eqref{eq.comp1} is more conservative than the one given by the RKA.

Next, a numerical illustration is implemented to compare the bounds given by \eqref{eq.comp1},
\eqref{eq.prodbd} and \eqref{eq.comp2}. The setup is as follows. Let $m= 30$ and $n= 3$, generate
$\AM = \text{randn}(30,3)$ and normalize each row to obtain $\BM$, generate $\xv =
\text{randn}(3,1)$ and compute $\yv = \BM \xv$. In the implementation of the RKA, we run
$1\thinspace 000$ realizations with the same initial value $\xv_0$ to obtain an average performance
result, which is reported in Fig.~1.

From the left panel in Fig.~1, we can see that the bound \eqref{eq.comp2} for characterizing the convergence of the RKA is closer to the real performance of the RKA, while the bounds given by \eqref{eq.comp1} and \eqref{eq.prodbd} for bounding the convergence of the KA are further away from the real performance of the KA.

The right panel in Fig.~1 shows a zoomed illustration of the bound given by \eqref{eq.comp1} and
\eqref{eq.prodbd}. We can observe that the bound given by~\eqref{eq.prodbd} improves upon~\eqref{eq.comp1}, which is enabled by the partitioning of the rows of the matrix.

\begin{figure}[htbp] 
\centering
\includegraphics[width=4.3in]{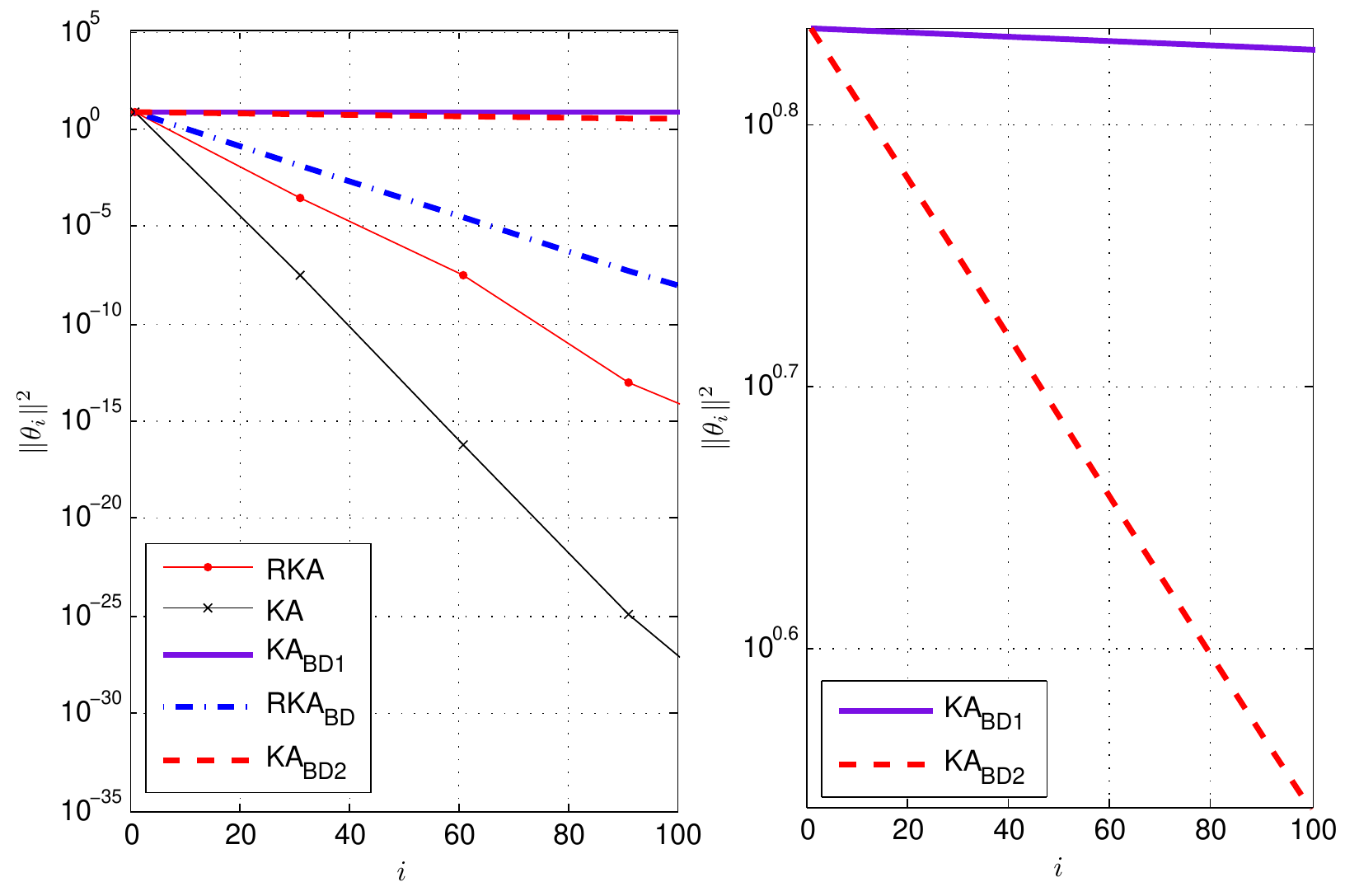}
\caption{In the left panel, the curves with tags '$\text{KA}$' and '$\text{RKA}$' illustrate the
  real performance of the KA and the RKA. The curves with tags '$\text{KA}_{\text{BD1}}$',
  '$\text{KA}_{\text{BD2}}$' and '$\text{RKA}_{\text{BD}}$' illustrate the bounds given by
  (\ref{eq.comp1}), (\ref{eq.prodbd}) and  (\ref{eq.comp2}), respectively. In the right panel, a zoomed illustration for the curves '$\text{KA}_{\text{BD1}}$' and '$\text{KA}_{\text{BD2}}$' in the left panel is given.}
\label{fig.rate}
\end{figure}
\subsection{Comparison with the bound given in \cite{c24}}
\label{sec:num}
To compare the result given in \cite{c24}, we assume that $\AM = \BM$, and that they are square and
invertible matrices.  Under these assumptions, the involved quantity $\mu$ in Corollary~4.2 of
\cite{c24} can be approximated by $\frac{1}{\sqrt{m}\|\BM^{\dagger}\|_2}$ given the results in
Theorem~2.2 of~\cite{c24}. Hence, the convergence rate of the KA given by Theorem~3.1 in~\cite{c24}
can be written as
\begin{align}
\rho^2 \le 1- \frac{\lambda(2-\lambda)}{m\left[1+(m-1)\lambda^2\right]\|\BM^{\dagger}\|_2^2},
\label{eq1}
\end{align}
where $\lambda \in (0,2)$. The result of the current work reads as
\begin{align}
\rho^2 \le 1- \frac{\lambda(2-\lambda)}{(2 + \lambda^2 m^2)\|\BM^{\dagger}\|_2^2},
\label{eq2}
\end{align}
where $\lambda \in (0,2)$. A closer look at the two bounds~\eqref{eq1} and~\eqref{eq2} reveals the
following:
\begin{enumerate}
\item
The optimal choice for the right hand side (RHS) of~\eqref{eq1} is $\lambda = \frac{\sqrt{4m-3}-1}{2(m-1)}$, resulting in $\rho^2 \leq 1
- \frac{2}{m(\sqrt{4m-3}+1)\|\BM^{\dagger}\|_2^2}$. When $m$ is large, $\rho^2$ decreases with the speed $\frac{1}{m^{1.5}\|\BM^{\dagger}\|_2^2}$.
\item
When $\lambda =
\frac{\sqrt{2}}{m}$ (a suboptimal choice for simplicity), \eqref{eq2} gives that  $\rho^2 \leq 1 - \frac{\sqrt{2}(2-\frac{\sqrt{2}}{m})}{4m\|\BM^{\dagger}\|_2^2}$. When $m$ is large, $\rho^2$ decreases in the speed of $\frac{1}{m\|\BM^{\dagger}\|_2^2}$, faster than the one in \cite{c24}. A comparison of both bounds when the optimal $\lambda$ are chosen is given in Fig. $\eqref{fig.rate2}$.

\item
When $\lambda$ is chosen to be 1, both bounds decrease in the order of $m^{-2}$.
\end{enumerate}
\begin{figure}[htbp] 
\centering
\includegraphics[width=4.3in]{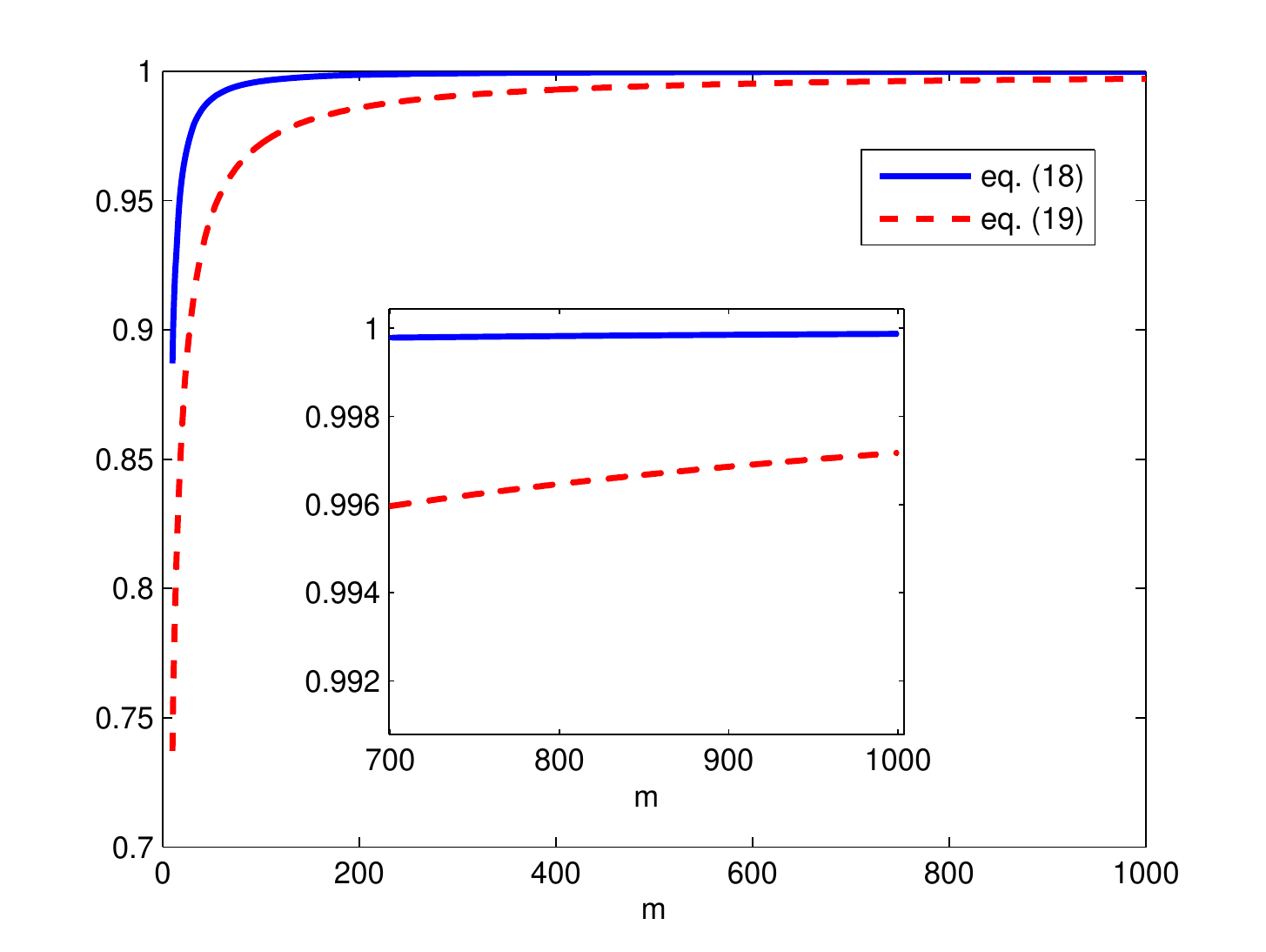}
\caption{The bounds \eqref{eq1} and \eqref{eq2} (using the optimal parameters) are plotted for the
  given $\|\BM^{\dagger}\|_2 = 0.5$ and $m$ ranging from $10$ to $1\thinspace 000$. The result shows
  that the bound proposed in this work is always lower than the one given in \cite{c24} under the
  experimental settings.}
\label{fig.rate2}
\end{figure}

\subsection{Comparison with the bound given in \cite{c26}}
\label{sec:nbd}
We will once again assume that each row of
$\AM$ is normalized, i.e. $\BM = \AM$. In \cite{c26} the authors makes use of a {\em subspace
  correction method} in studying the convergence speed of the KA. They show that (see eq.~(31) in
\cite{c26}), when the best relaxation parameter $\lambda$ is chosen, $\rho^2$ can be bounded from
above according to
\begin{align}
  \rho^2 \le 1-\frac{1}{\lfloor\log_2(2m)\rfloor\|\BM\|_2^2 \|\BM^{\dagger}\|_2^2}.
  \label{eq.newbd}
\end{align}
As we discussed in the previous section, when a near-optimal $\lambda$ (i.e. $\lambda$ is chosen as
$\frac{\sqrt{2}}{m}$) is used, the upper bound implied by our analysis gives that
\begin{align}
  \rho^2 \leq 1 - \frac{\sqrt{2}(2-\frac{\sqrt{2}}{m})}{4m\|\BM^{\dagger}\|_2^2}.
  \label{eq.ourbd}
\end{align}
By assumption we have $\|\BM\|_F^2 = m$, which implies that
\begin{align}
  \|\BM\|_2^2 = \|\BM^{T}\BM\|_2 \ge \frac{\text{tr}(\BM^{T}\BM)}{n} = \frac{\|\BM\|_F^2}{n} = \frac{m}{n}.
\end{align}
Hence, the bound obtained by~\cite{c26} will decrease with a speed of
$\frac{1}{m\log_2(m)\|\BM^{\dagger}\|_2^2}$ as $m$ increases, while the present work gives the decreasing speed of
$\frac{1}{m\|\BM^{\dagger}\|_2^2}$.
\end{section}

\begin{section}{Summary}
\label{sec:Conc}
By studying the stability property of a time-varying dynamical system that is related to the KA we
have been able to establish some new results concerning the convergence speed of the algorithm. The
new results are also compared to several related, previously available results. Let us end the
discussion by noting that the following two ideas can possibly lead to further improvements of the
results. One potential idea is trying to improve the inequality in \eqref{eq.loose}, since this part
introduces much of the approximations in establishing the main result of the note; another idea is
to try to find an optimal partitioning of the rows of the matrix~$\AM$, such that the right hand
side of~\eqref{eq.prodbd} is minimized.
\end{section}

\section{Acknowledgement}
The authors would like to thank the reviewers for useful comments and pointing out the references~\cite{c24,c26}, and Marcus Björk for helpful suggestions.

\end{document}